\documentclass[11pt]{article}
\usepackage[margin=1in]{geometry}
\usepackage{palatino}
\usepackage{amsthm,amsmath,amssymb,amsfonts}
\usepackage[final]{graphicx}
\usepackage{subcaption}
\usepackage{hyperref}
\usepackage{booktabs}
\usepackage{algorithm, algpseudocode}
\usepackage{authblk}
\usepackage{xcolor, colortbl}
\usepackage{ifdraft}
\usepackage[utf8]{inputenc}
\usepackage{xspace}
\usepackage[final]{pdfpages}
\usepackage{graphicx}
\usepackage{verbatim}
\usepackage{xfrac}
\usepackage{array, wrapfig}
\usepackage{multirow}

\newcommand{\dhcjoin}{S-DHC\xspace}

\usepackage{fancyhdr}
\fancypagestyle{empty}{
   \fancyhf{} 
   
}

\theoremstyle{definition}
\newtheorem{definition}{Definition}

\theoremstyle{theorem}
\newtheorem{theorem}{Theorem}
\newtheorem{lemma}{Lemma}
\newtheorem{corollary}{Corollary}
\newtheorem{proposition}{Proposition}

\newcommand{\D}{\mathcal{D}}

\newcommand{\eat}[1]{}

\graphicspath{{screenshots/}{figs/}}

\title{PHSafe: Disclosure Avoidance for the 2020 Census Supplemental Demographic and Housing Characteristics File (S-DHC)}
\author[1]{William Sexton}
\author[1]{Skye Berghel}
\author[1]{Bayard Carlson}
\author[1]{Sam Haney}
\author[1]{Luke Hartman}
\author[1]{Michael Hay}
\author[1]{Ashwin Machanavajjhala}
\author[1]{Gerome Miklau}
\author[1]{Amritha Pai}
\author[1]{Simran Rajpal}
\author[1]{David Pujol}
\author[1]{Ruchit Shrestha}
\author[1]{Daniel Simmons-Marengo}

\affil[1]{Tumult Labs}
\date{May 28, 2024}

\begin{document}

\maketitle
\begin{abstract}
This article describes the disclosure avoidance algorithm that the U.S. Census Bureau used to protect the 2020 Census Supplemental Demographic and Housing Characteristics File (\dhcjoin). The tabulations contain statistics of counts of U.S. persons living in certain types of households, including averages. The article describes the PHSafe algorithm, which is based on adding noise drawn from a discrete Gaussian distribution to the statistics of interest. We prove that the algorithm satisfies a well-studied variant of differential privacy, called zero-concentrated differential privacy. We then describe how the algorithm was implemented on Tumult Analytics and briefly outline the parameterization and tuning of the algorithm.  
\end{abstract}

\newpage
\tableofcontents
\newpage

\section{Introduction}
The U.S. Census Bureau (Census Bureau) conducts a census of the U.S. population every 10 years. The 2020 Census is the most recent such undertaking. As part of the 2020 Census, the Census Bureau manages the collection, processing, and publication of census data to the U.S. people. While handling these data, the Census Bureau is obligated to preserve the confidentiality of census respondents. To achieve this goal, the Census Bureau utilizes a Disclosure Avoidance System (DAS) that incorporates privacy algorithms into the data processing and dissemination procedures of the 2020 Census. For the 2020 Census, the DAS was rebuilt to improve its privacy protection mechanisms. The 2020 DAS adheres to a privacy framework known as \emph{differential privacy}. This framework admits many different variants or instances of privacy definitions. The most well-known instance of the framework is called ``pure'' differential privacy. However, in this paper, we assume that the term ``differential privacy'' refers to the variant known as zero-concentrated differential privacy, unless otherwise specified.  Data publication mechanisms that fit the definition of any instance of the differential privacy framework achieve mathematically provable guarantees about the privacy loss incurred by data publication.  

The DAS leverages a slate of differential privacy mechanisms to protect the Census Bureau's various data products for the 2020 Census. These data products include the 2020 Census Redistricting Data (P.L. 94-171) Summary File, Demographic and Housing Characteristics File (DHC), Demographic Profile, Detailed Demographic and Housing Characteristics File A (Detailed DHC-A), Detailed Demographic and Housing Characteristics File B (Detailed DHC-B), and Supplemental Demographic and Housing Characteristics File (S-DHC). Since each data product has its own challenges regarding confidentiality protection, the DAS customizes its privacy algorithms to optimize protection and accuracy on a product-by-product basis. Some closely related data products use the same algorithm. 

The focus of this paper is PHSafe, an algorithm designed specifically to provide differential privacy guarantees for the production and release of the S-DHC.  

There are three main goals of this article:

\begin{enumerate}
    \item Describe the PHSafe algorithm and how it meets the requirements of the S-DHC.
    \item Prove the privacy properties of the PHSafe algorithm.
    \item Describe the parameters in the PHSafe algorithm and how they impact privacy-accuracy trade-offs.
\end{enumerate}

The paper is organized as follows. Section \ref{sec:problem} covers a full description of the data product and privacy release problem. Section \ref{sec:phsafe-algorithm} provides a pseudocode description of PHSafe that details how privacy protection is applied to create the S-DHC tabular summaries. Sections \ref{sec:privacy-prelim} and \ref{sec:privacy-analysis} provide relevant background material on the differential privacy framework and explain how PHSafe adheres to the framework, respectively. In Section \ref{sec:implementation}, we discuss the differences between our pseudocode abstraction and the programmed codebase of the algorithm. Finally, Section \ref{sec:params} describes the parameters in PHSafe that impact the privacy-accuracy trade-offs of the algorithm. 

\section{Problem Setup}\label{sec:problem}
The S-DHC includes counts of people in the United States and Puerto Rico living in certain types of households, including averages. It includes eight published tables--six of which are repeated by major race and ethnicity. A common feature in most of these tabulations is that statistics are generated after \textit{joining} person characteristics (e.g., age) with properties of the households they live in (e.g., family type of the household). In this section, we define relevant concepts, outline the statistics to be released, and then formulate the differentially private algorithm design problem. 

\subsection{Geography}
A geographic summary level is a set of geographic areas with nonoverlapping boundaries, such as the set of all states. The S-DHC contains statistics for two geographic summary levels:

\begin{itemize}
    \item Nation
    \item State or State equivalent
\end{itemize}

Washington, D.C. is an example of a State equivalent. However, we tend to omit the ``or State equivalent'' qualifier in the remainder of this document. Throughout this paper, we adopt the convention of capitalizing levels and using lowercase for entities within a level (e.g., the state of Maine is included in the State summary level).

\subsection{Race and Ethnicity}
\label{sec:race-defs}

Every person is associated with one or more \textit{major} race categories: White, Black or African American, American Indian or Alaska Native, Asian, Native Hawaiian or Other Pacific Islander, and Some Other Race. Every person also has a binary ethnicity indicator (Hispanic or Latino, Not Hispanic or Latino). 

In the collected 2020 Census data, every household has exactly one person that is designated as its householder. Every household is associated with the major races and ethnicity of its householder.

Six of the S-DHC tables are repeated by the following race and ethnicity iterations: 

\begin{itemize}
    \item[(*)] Total population
    \item[(A)] White alone
    \item[(B)] Black or African American alone
    \item[(C)] American Indian and Alaska Native alone
    \item[(D)] Asian alone
    \item[(E)] Native Hawaiian and Other Pacific Islander alone
    \item[(F)] Some Other Race alone
    \item[(G)] Two or more major races
    \item[(H)] Hispanic or Latino
    \item[(I)] White alone, not Hispanic or Latino
\end{itemize}

These iterations give rise to three distinct iteration levels for which the S-DHC statistics may be released.

\begin{itemize}
    \item \textit{Unattributed}: These statistics represent the total population; they are not constrained by either race or ethnicity. This level only has one iteration, denoted as the * iteration in the above list.  
    \item \textit{A-G}: These statistics are categorized by exactly one of the major race iterations denoted by A, B, C, D, E, F, or G. No constraint is placed on ethnicity.  
    \item \textit{H-I}: These statistics are categorized by either the H iteration or the I iteration. H is constrained by ethnicity and I is constrained both race and ethnicity.
\end{itemize}

Depending on the specific statistic, the iterations may be applied to either individuals or households (i.e., to the householder). For example, one statistic may count the number of people living in households who are under 18 years old and White alone, while another statistic may count the number of people under 18 years old that are living in a household where the race of the householder is White alone.  

Regardless of whether the iteration is applied to each individual or to the household, the categories within a level are mutually exclusive. For example, in the A-G level, an individual or household cannot simultaneously be classified as White alone (A) and Black or African American alone (B). However, individuals or households can be classified in multiple iterations across levels. For example, an individual or household can be classified as Unattributed (*), Asian alone (D), and Hispanic or Latino (H).

\subsection{Population Groups}\label{sec:population-group-levels}

A \textit{population group} is a pair $(g, c)$, where $g$ is a geographic entity (e.g., the state of Wisconsin or the nation) and $c$ is a race or ethnicity iteration (e.g., Asian alone). Population groups are divided into \textit{population group levels}. We often identify a population group level by specifying a (geography level, characteristic iteration level) pair. Each population group level is a set of population groups, where the population group's geographic entity belongs to the specified geography level and its characteristic iteration belongs to the specified characteristic iteration level. More formally, the S-DHC requires the publication of statistics for the following population group levels:

\begin{itemize}
    \item (Nation, Unattributed) $\equiv$ $\{(g, c): g \text{ is the nation}, c \text{ is the * iteration}\}$
    \item (Nation, A-G) $\equiv$ $\{(g, c): g \text{ is the nation}, c \in \{A, B, C, D, E, F, G\}\}$
    \item (Nation, H-I) $\equiv$ $\{(g, c): g \text{ is the nation)}, c \in \{H, I\}\}$
    \item (State, Unattributed) $\equiv$ $\{(g, c): g \text{ is a state}, c \text{ is the * iteration}\}$
    \item (State, A-G) $\equiv$ $\{(g, c): g \text{ is a state}, c \in \{A, B, C, D, E, F, G\}\}$
    \item (State, H-I) $\equiv$ $\{(g, c): g \text{ is a state}, c \in \{H, I\}\}$
\end{itemize}

A person or household is associated with at most one population group in the set that comprises a population group level. For example, a household in Texas may be associated with H or I but not both, since a householder cannot be both Hispanic or Latino and not Hispanic or Latino. Assuming the householder is Hispanic or Latino, the household would be associated with the (Texas, H) population group in the (State, H-I) population group level. It is also possible for a household to be associated with neither H nor I if, for example, the householder was Native Hawaiian and Other Pacific Islander alone and not Hispanic or Latino. 

One person or household may belong to population groups across multiple population group levels. For example, a Californian household with a Native Hawaiian and Other Pacific Islander alone, not Hispanic or Latino householder would be associated with the (nation, *), (nation, E), (CA, *), and (CA, E) population groups within the levels (Nation, Unattributed), (Nation, A-G), (State, Unattributed), and (State, A-G), respectively.

\subsection{Supplemental Demographic and Housing Characteristics File}

The \dhcjoin aims to tabulate statistics by population groups. 
The goal is to release the following eight tables:

\begin{itemize}
    \item[(PH1):] Average Household Size by Age (Table \ref{tab:ph1-shell}).
    \item[(PH2):] Household Type for the Population in Households (Table \ref{tab:ph2-shell}).
    \item[(PH3):] Household Type by Relationship for the Population Under 18 Years (Table \ref{tab:ph3-shell}).
    \item[(PH4):] Population in Families by Age (Table \ref{tab:ph4-shell}).
    \item[(PH5):] Average Family Size by Age (Table \ref{tab:ph5-shell}). 
    \item[(PH6):] Family Type and Age for Own Children Under 18 Years (Table \ref{tab:ph6-shell}).
    \item[(PH7):] Total Population in Occupied Housing Units by Tenure (Table \ref{tab:ph7-shell}).
    \item[(PH8):] Average Household Size of Occupied Housing Units by Tenure (Table \ref{tab:ph8-shell}).
\end{itemize}

PH1, PH3, PH4, PH5, PH7 and PH8 are released for all population group levels whereas PH2 and PH6 are only released for (Nation, Unattributed) and (State, Unattributed). PH3 iterates by the race and ethnicity of each person, whereas PH1, PH4, PH5, PH7 and PH8 iterate by the race and ethnicity of the householder (or household).

Each table has a \emph{basis}, a set of fine-grained, disaggregated table cells from which the remaining cells can be generated. In the below table shells, the dark text indicates which cells form the table's basis, while the light text shows the aggregated table cells. For a given table, every in-universe person can be assigned to exactly one category of the table's basis. For example, with PH4, every person in a family is exclusively either under 18 years or 18 years and over.

\begin{table}[H]
\begin{tabular}{l}
\textbf{PH1: Average Household Size by Age}\\
\textit{Universe: Households.}\\
\textcolor{lightgray}{Total:}\\
\quad \quad Under 18 years\\
\quad \quad 18 years and over\\
\end{tabular}
\caption{\label{tab:ph1-shell}This table contains average household size by age.}
\end{table}

\begin{table}[H]
\begin{tabular}{l}
\textbf{PH2: Household Type for the Population in Households}\\
\textit{Universe: Population in households.}\\
\textcolor{lightgray}{Total:}\\
\quad \quad \textcolor{lightgray}{In married couple household:}\\
\quad \quad \quad \quad Opposite-sex married couple\\
\quad \quad \quad \quad Same-sex married couple\\
\quad \quad \textcolor{lightgray}{In cohabiting couple family:}\\
\quad \quad \quad \quad Opposite-sex cohabiting couple\\
\quad \quad \quad \quad Same-sex cohabiting couple\\
\quad \quad \textcolor{lightgray}{Male householder, no spouse or partner present:}\\
\quad \quad \quad \quad Living alone\\
\quad \quad \quad \quad Living with others\\
\quad \quad \textcolor{lightgray}{Female householder, no spouse or partner present:}\\
\quad \quad \quad \quad Living alone\\
\quad \quad \quad \quad Living with others\\
\end{tabular}
\caption{\label{tab:ph2-shell}This table contains household population counts based on the marital and cohabitation status of the householder.}
\end{table}

\begin{table}[H]
\begin{tabular}{l}
\textbf{PH3: Household Type by Relationship for the Population Under 18 Years}\\
\textit{Universe: Population in households under 18 years.}\\
\textcolor{lightgray}{Total:}\\
\quad \quad Householder, spouse, unmarried partner, or nonrelative\\
\quad \quad \textcolor{lightgray}{Own child:}\\
\quad \quad \quad \quad In married couple family\\
\quad \quad \quad \quad In cohabiting couple family\\
\quad \quad \quad \quad In male householder, no spouse or partner present family\\
\quad \quad \quad \quad In female householder, no spouse or partner present family\\
\quad \quad \textcolor{lightgray}{Other relatives:}\\
\quad \quad \quad \quad Grandchild\\
\quad \quad \quad \quad Other relatives\\
\end{tabular}
\caption{\label{tab:ph3-shell}This table contains counts for children under 18 years based on their household type.}
\end{table}

\begin{table}[H]
\begin{tabular}{l}
\textbf{PH4: Population in Families by Age}\\
\textit{Universe: Population in families.}\\
\textcolor{lightgray}{Total:}\\
\quad \quad Under 18 years\\
\quad \quad 18 years and over\\
\end{tabular}
\caption{\label{tab:ph4-shell}This table contains population in families by age.}
\end{table}

\begin{table}[H]
\begin{tabular}{l}
\textbf{PH5: Average Family Size by Age}\\
\textit{Universe: Families.}\\
\textcolor{lightgray}{Total:}\\
\quad \quad Under 18 years\\
\quad \quad 18 years and over\\
\end{tabular}
\caption{\label{tab:ph5-shell}This table contains average family size by age.}
\end{table}

\begin{table}[H]
\begin{tabular}{l}
\textbf{PH6: Family Type and Age for Own Children Under 18 Years}\\
\textit{Universe: Own Children Under 18 Years.}\\
\textcolor{lightgray}{Total:}\\
\quad \quad \textcolor{lightgray}{In married couple household:}\\
\quad \quad \quad \quad Under 4 years\\
\quad \quad \quad \quad 4 and 5 years\\
\quad \quad \quad \quad 6 to 11 years\\
\quad \quad \quad \quad 12 and 17 years\\
\quad \quad \textcolor{lightgray}{In cohabiting couple family:}\\
\quad \quad \quad \quad Under 4 years old\\
\quad \quad \quad \quad 4 and 5 years\\
\quad \quad \quad \quad 6 to 11 years\\
\quad \quad \quad \quad 12 and 17 years\\
\quad \quad \textcolor{lightgray}{Male householder, no spouse or partner present family:}\\
\quad \quad \quad \quad Under 4 years\\
\quad \quad \quad \quad 4 and 5 years\\
\quad \quad \quad \quad 6 to 11 years\\
\quad \quad \quad \quad 12 and 17 years\\
\quad \quad \textcolor{lightgray}{Female householder, no spouse or partner present family:}\\
\quad \quad \quad \quad Under 4 years\\
\quad \quad \quad \quad 4 and 5 years\\
\quad \quad \quad \quad 6 to 11 years\\
\quad \quad \quad \quad 12 and 17 years\\
\end{tabular}
\caption{\label{tab:ph6-shell}This table contains counts of own children under 18 years by family type and age of children.}
\end{table}

\begin{table}[H]
\begin{tabular}{l}
\textbf{PH7: Total Population in Occupied Housing Units by Tenure}\\
\textit{Universe: Population in Occupied Housing Units}\\
\textcolor{lightgray}{Total:}\\
\quad \quad Owned with a mortgage or a loan\\
\quad \quad Owned free and clear\\
\quad \quad Renter occupied\\
\end{tabular}
\caption{\label{tab:ph7-shell} This table contains household population counts based on the tenure status of the occupied housing unit.}
\end{table}

\begin{table}[H]
\begin{tabular}{l}
\textbf{PH8: Average Household Size of Occupied Housing Units by Tenure}\\
\textit{Universe: Occupied Housing Units}\\
Total:\\
\quad \quad Owner occupied\\
\quad \quad Renter occupied\\
\end{tabular}
\caption{\label{tab:ph8-shell} This table contains average household size based on the tenure status of the occupied housing unit. The PH8 total does not use light text because the indented cells are not expected to add up to the total.}
\end{table}

We end this section with some comments on the average tables PH1, PH5, and PH8. The PH1 table shows a hierarchical structure, with the age breakouts nested below the total. This accurately implies the nested averages add up to the total average--a consequence of each cell in PH1 sharing the same denominator, the total number of households. The PH5 table is similar. The nested averages add up to the total average because every cell uses the total number of families as a denominator. The PH8 table does not share this similarity. Despite displaying a hierarchical structure, the table behaves non-hierarchically because each cell uses a different denominator (total number of occupied housing units vs. total number of owner occupied housing units vs. total number of renter occupied housing units), which results in the total average not equaling the sum of the owner occupied average and the renter occupied average.

\subsection{Privacy Release Problem}

The Census Bureau is required to abide by the regulations of Title 13 when managing the collection, storage, and release of statistics about persons and households in the United States \cite{title13}. Furthermore, legacy statistical disclosure limitation techniques fail to uphold the Census Bureau's privacy standards since they are vulnerable to attacks that can reconstruct sensitive records from aggregate statistics \cite{ap-census-attack}. Hence, the Census Bureau decided to release many of the 2020 Census data products, including the S-DHC, using algorithms that satisfy modern privacy definitions like differential privacy \cite{dsep-dp}.

In this paper, we describe PHSafe, a differentially private algorithm for releasing the statistics that make up the \dhcjoin. PHSafe was designed to satisfy the following criteria:

\begin{itemize}
    \item \textit{Privacy:} The algorithm must satisfy zCDP with respect to arbitrary changes of any person's record values. This requires a careful accounting of privacy loss, as changing the values of one person's record, including which household they belong to, can change the properties of at most two households, and consequently manifest a large change in the output statistics. 
    \item \textit{Population Groups:} The algorithm must release statistics for a predefined set of race and ethnicity iterations and the Nation and State geography levels. 
    \item \textit{Static tabulations:}  As shown previously, tables are arranged with a hierarchical table cell structure. PHSafe directly estimates the most detailed level of disjoint table cells (i.e., the table basis) for a table. For example, for table PH4, PHSafe takes noisy measurements of the age breakouts but not the total.

    \item \textit{Averages and Statistical Postprocessing:} The PHSafe algorithm must output independent noisy measurements of the numerators and denominators associated with the average tabulations. A separate statistical postprocessing model combines these to produce the average tables. Statistical postprocessing is out of scope for this paper.

    \item \textit{Accuracy:} The algorithm must satisfy pre-determined accuracy targets, defined in terms of the 90\% margins of error (MOE) on the output counts. The margins of error capture error due to noise infusion within PHSafe but do not capture other sources of error, whether internal or external to PHSafe. Examples of other error sources include truncation within PHSafe and undercounting in the 2020 Census data collection procedures. Accuracy targets vary for different tables and population group levels. The desired targets are outlined in Table \ref{tab:moe-targets}.
    
    \item \textit{Integrality:} The output statistics must be integers.
    
    \item \textit{Minimal Consistency:} For the most part, the  PHSafe differential privacy algorithm is not required to ensure consistency.\footnote{Special handling of the average tables results in some consistency between PH4 and PH5, as well as between PH7 and PH8. This is discussed in more detail in Section \ref{sec:phsafe-algorithm}.} That is, different counts output by the algorithm need not be consistent with each other. One example of an expected inconsistency would be the number of people in households in the United States not equaling the sum of the number of people in households across all states. Another example would be different estimates of the same statistic appearing in different tables. For example, the population in households in PH2 is not expected to equal the population in households in PH7. 
    
    \item \textit{Negativity and Statistical Postprocessing:} The output of the PHSafe algorithm may contain negative counts. The previously mentioned statistical postprocessing model ensures non-negativity of counts and averages. 
\end{itemize}

In the rest of the paper, we describe the PHSafe differential privacy algorithm, discuss implementation and parameter tuning, and analyze bounds on the privacy loss achievable while satisfying the constraints mentioned above.

\section{PHSafe Algorithm}\label{sec:phsafe-algorithm}
PHSafe is a privacy algorithm for releasing tabulations pertaining to the population in households from the 2020 Census, iterated by major race and ethnicity at the Nation and State geography levels. This section covers a simplified abstraction of the PHSafe algorithm. Additional implementation details are discussed in Section \ref{sec:implementation}. In this section, we describe the algorithm as applied to the United States. Puerto Rico is discussed in Section \ref{sec:pr}. The algorithm acts on a selection of private dataframes derived from the 2020 Census.

\subsection{Input Data Description}\label{sec:input-descriptions}
The input dataframes for PHSafe are sourced from the Census Edited File (CEF). The CEF contains person and household attributes stored in a relational database. Person records are linkable to household records via a Master Address File ID (MAFID) join key. Many of the attributes available in the CEF do not factor into the S-DHC tabulations. Therefore, we assume a reduced-form data extraction as inputs to the PHSafe algorithm. These inputs are detailed below.

\subsubsection{Person Dataframe}
The first input is a dataframe, denoted person\_df, containing one row for each person in the United States. Each row consists of the following attributes: StateID, MAFID, Age, RaceEth, and Relationship.

\vspace{\baselineskip}

\noindent \textbf{StateID} is a single attribute that geolocates a person record to a unique state. We note that all records in the United States are vacuously included in the nation geographic entity. 

\vspace{\baselineskip}

\noindent \textbf{MAFID} is a foreign key that links a person record to a household record.

\vspace{\baselineskip}

\noindent \textbf{Age} is an integer which represents the age (in years) of the person.

\vspace{\baselineskip}

\noindent \textbf{RaceEth} is a single attribute that encodes up to six major race categories and a binary ethnicity indicator of the person. That is, one person's RaceEth attribute may indicate the person has one race (e.g., White alone) and is Hispanic or Latino while another person's RaceEth attribute indicates the person has two races (e.g., White and Black or African American) and is Not Hispanic or Latino. 

\vspace{\baselineskip}
    
\noindent \textbf{Relationship:} An attribute that represents the person's relationship to the householder for their associated household. Examples include opposite-sex spouse, grandchild, or nonrelative.

\subsubsection{Housing Unit Dataframe (Unit Dataframe)}
The second input dataframe, called unit\_df, contains one row for each occupied housing unit in the United States.\footnote{Vacant housing units exist in the CEF, but they do not factor into the S-DHC, so we assume an input extraction that only contains occupied housing units.} Each row in the unit dataframe consists of the following attributes: StateID, MAFID, HouseholderRaceEth, Tenure, HouseholdType.  

\vspace{\baselineskip}

\noindent \textbf{StateID} is a single attribute that geolocates the household record to its unique state. We note that all records are vacuously included in the nation geographic entity.

\vspace{\baselineskip}

\noindent \textbf{MAFID} is a primary key that uniquely identifies the household. 

\vspace{\baselineskip}

\noindent \textbf{HouseholderRaceEth} encodes the race and ethnicity of the householder in a similar manner as the RaceEth attribute in the person dataframe. 

\vspace{\baselineskip}

\noindent \textbf{Tenure} encodes the occupied housing unit's tenure status as either owned with a mortgage, owned free and clear, or rented.

\vspace{\baselineskip}

\noindent \textbf{HouseholdType} defines the household's composition, capturing key relationships between the householder and other members of the household. Examples include opposite-sex married couple family, same-sex cohabiting couple family, and male householder living alone.

\subsection{The Algorithm Description}\label{sec:algorithm-description}

We present an abstraction of the PHSafe implementation. Notation is outlined in Table \ref{tab:algorithm-notation}. PHSafe is the privacy algorithm for releasing population in household statistics from the 2020 Census for the S-DHC tables (PH1, PH2, PH3, PH4, PH5, PH6, PH7, and PH8) iterated by race and ethnicity as applicable. In our abstraction, PHSafe is executed on each table separately, with some special consideration for the average tables. 

Tables PH1, PH5, and PH8 are average tables. PHSafe does not directly release averages, but rather reports noisy numerator and denominator count estimates independently. We denote these by appending the suffixes \emph{\_num} and \emph{\_denom} (e.g., PH1\_num and PH1\_denom). We note that PH4 is identical to PH5\_num. To avoid redundancy, PHSafe only estimates PH4 and omits PH5\_num from its processing. Furthermore, we observe that PH8\_num and PH8\_denom do have hierarchical table structures, despite PH8 behaving non-hierarchically. For PH8\_num, the population in owner occupied units and the population in renter occupied units does sum to the total population in occupied housing units. For PH8\_denom, the number of owner occupied units and the number of renter occupied units sums to the total count of occupied housing units. Next, we note that PH8\_num is derivable from PH7, since the population in owner occupied units equals the sum of the populations in owned with a mortgage or a loan units and owned free and clear units. With this in mind, PHSafe opts to not directly estimate PH8\_num, but rather to derive its estimates from PH7. PH8\_denom is directly estimated.

Thus, we assume the PHSafe mechanism requires nine independent runs to generate the S-DHC: PH1\_num, PH1\_denom, PH2, PH3, PH4, PH5\_denom, PH6, PH7, and PH8\_denom.

The algorithm acts on the pair of private dataframes (person\_df, unit\_df) to produce its output for population groups. Each table defines a set of hierarchical data cells such as ``total'', ``married couple household'', or ``own child'' with respect to a table universe, the entities being tabulated. For example, the table universe for PH3 is the population under 18 years in households. The PHSafe algorithm only produces estimates for the basis of a table, meaning that records can be partitioned into disjoint categories that add up to the table universe total. PHSafe varies from table to table but conceptually it takes one of two forms:

\begin{enumerate}
    \item For PH1\_num, PH2, PH3, PH4, PH6, and PH7, PHSafe uses a filter-join-transform-measure query structure, described below in detail. The query involves a data transformation and a noisy measurement. Algorithm \ref{alg:phsafe-phjoin} presents the pseudocode.
    \item For PH1\_denom, PH5\_denom, and PH8\_denom, PHSafe uses a filter-transform-measure query structure. Algorithm \ref{alg:phsafe-unit-count} presents the pseudocode.
\end{enumerate}

In the first form (shown as Algorithm \ref{alg:phsafe-phjoin}), the goal is to count people in categories distinguished by household and person attributes. We use the table PH1\_num as an example. PH1\_num counts the population in households by voting age. These counts are produced for the Nation and State geography levels. The table is also iterated by the race and ethnicity of the householder for all iteration levels described in Section \ref{sec:problem}.  The filter-join-transform-measure structure is applied to all population group levels for table PH1\_num. For each population group level, there are four steps:

\begin{itemize}
    \item[\textbf{Filter:}] The person\_df and unit\_df are filtered to remove out-of-universe records. The universe of table PH1\_num is the population in households. Given the inputs described in Section \ref{sec:input-descriptions}, no records are removed in the case of PH1\_num. However, for tables like PH3, the filter would remove people 18 years or over.
    \item[\textbf{Join:}] The person\_df and unit\_df are inner joined on MAFID. The join internally enforces a truncation threshold that limits the number of people within a household to an input threshold (refer to Algorithm \ref{alg:trunc-and-join} for details of the join).
    \item[\textbf{Transform:}] Records in the joined dataframe are mapped to at most one population group within the current level. For example, for table PH1\_num, if the level is (State, A-G), each person in a household is mapped to its associated population group based on the state they reside in and the major race iteration (White alone, Black or African American alone, etc.) of the householder. Next, the person is mapped to exactly one basis table cell based on the criteria specified within the table specifications. For PH1\_num, the basis table cells are either ``under 18 years old'' or ``18 years and older.''
    \item[\textbf{Measure:}] A noisy measurement count (Algorithm \ref{alg:vector-base-discrete-gaussian}) is generated for each combination of population group and basis table cell.
\end{itemize}

The second form (Algorithm \ref{alg:phsafe-unit-count}) is similar to the first with a few exceptions. The measurement is a count of occupied housing units in categories distinguished by household attributes alone, so only the unit\_df is used. Thus, the join operator is not needed. 

\begin{table}[H]
    \centering
    \begin{tabular}{c p{.8\linewidth}}
        \toprule
        \textbf{Notation} & \textbf{Description} \\
        \midrule
        $\omega$ & The number of population group levels. \\
        $\mathcal{P}_i$ & Population group level $i$.  \\
        $\rho_i$ & The privacy-loss budget allocated to population group level $i$. \\
        $g_i$ & A 1-stable mapping of records to the population group in $\mathcal{P}_i$ to which the record belongs. \\
        $f$ & A 1-stable mapping of records to a set of table basis cells. This mapping varies by table. \\
        $\tau$ & The truncation threshold. \\
    \end{tabular}
    \caption{A summary of the notation used in Section~\ref{sec:phsafe-algorithm}}
    \label{tab:algorithm-notation}
\end{table}

\begin{algorithm}[H]
\caption{\label{alg:phsafe-phjoin} The main PHSafe algorithm for PH1\_num, PH2, PH3, PH4, PH6, PH7}
\begin{algorithmic}[1]

\Require person\_df: Private dataframe with attributes [StateID, MAFID, Age, RaceEth, Relationship] and one row for each person in the United States.
\Require unit\_df: Private dataframe with attributes [StateID, MAFID, HouseholderRaceEth, Tenure, HouseholdType] and one row for each occupied housing unit in the United States.
\Require $\{\rho_i\}_{i \in [1,\omega]}$: Privacy parameters for each population group level $i \in [1, \omega]$.
\Require $\tau$: Truncation threshold

\Procedure{PHSafe}{$person\_df$, $unit\_df$, $\{\rho_i\}$, $\tau$}
\State $person\_df \gets person\_df.$filter(PERSON\_CONDITION) \Comment{condition differs by table}
\State $unit\_df \gets unit\_df.$filter(UNIT\_CONDITION) \Comment{condition differs by table}

\State \textit{//Add attributes from unit\_df to rows from person\_df as shown in Algorithm \ref{alg:trunc-and-join}.}
\State $df \gets person\_df.$truncate\_and\_join$(unit\_df, MAFID, \tau)$

\For{$i \in [1, \omega]$} \label{line:pop-group-level-loop}
\State $df_i \leftarrow df$.map($g_i$) \label{line:query-start}
\Comment{Map each row to the appropriate population group}

\State $V \gets [ \ ] $
\For{$P \in \mathcal{P}_i$ }\label{line:df-group-loop-denom}
\State $v_i \gets$ \Call{VectorizePopulationGroup}{$df_i$, $P$,$f$} \Comment{Count vector for population group}
\State $V$.append($v_i$)
\Comment{Append to the vector of all counts}
\EndFor \label{line:query-end}
\State \Call{VectorDiscreteGaussian}{$V$, $\rho_i$, $(2\tau+2)$} \label{line:count}

\EndFor
\EndProcedure
\end{algorithmic}
\end{algorithm}

\begin{algorithm}[H]
\caption{\label{alg:phsafe-unit-count} The main PHSafe algorithm for PH1\_denom, PH5\_denom, PH8\_denom}
\begin{algorithmic}[1]
\Require Unit\_df: Private dataframe with attributes [StateID, MAFID, HouseholderRaceEth, Tenure, HouseholdType] and one row for each occupied housing unit in the United States. 
\Require $\{\rho_i\}_{i \in [1,\omega]}$: Privacy parameters for each population group level $i \in [1, \omega]$.

\Procedure{PHSafe}{$unit\_df$, $\{\rho_i\}$}
\State $unit\_df \gets unit\_df.$filter(UNIT\_CONDITION) \Comment{condition differs by table}

\For{$i \in [1, \omega]$} \label{line:pop-group-level-loop-unit-count}
\State $df_i \leftarrow unit\_df$.map($g_i$) \label{line:denom-query-start}
\Comment{Map each row to the appropriate population group}

\State $V \gets [ \ ] $
\For{$P \in \mathcal{P}_i$}\label{line:df-group-loop}
\State $v_i \gets$ \Call{VectorizePopulationGroup}{$df_i$, $P$,$f$} \Comment{Count vector for population group}
\State $V$.append($v_i$)
\Comment{Append to the vector of all counts}
\EndFor \label{line:denom-query-end}
\State \Call{VectorDiscreteGaussian}{$V$, $\rho_i$, $2$} \label{line:denom-count}
\EndFor
\EndProcedure
\end{algorithmic}
\end{algorithm}

\begin{algorithm}[H]
\caption{\label{alg:trunc-and-join} Truncate and join operator}
\begin{algorithmic}[1]
\Require{$left\_df$: A dataframe with attribute A.}
\Require{$right\_df$: Another dataframe with attribute A.}
\Require{A: An attribute that serves as the join key between two dataframes.}
\Require{$\tau$: A truncation threshold}
\Procedure{left\_df.truncate\_and\_join}{$right\_df, A, \tau$}

\For{$a \in domain(A)$}
\Comment{The domain of $A$ is taken from left\_df}
\State $\mathcal{R} \gets \{row \in left\_df | row[A] = a\}$
\State $left\_df \gets left\_df.drop(\mathcal{R})$
\If{$|\mathcal{R}| > \tau$}
\State \textit{// In particular, we order records first by a hash function (to avoid bias that would result from ordering by record attributes) and second by record attributes (in the unlikely scenario that two records share the same hash).}
\State $\mathcal{R} \gets$ the first $\tau$ elements of $\mathcal{R}$ according to a predefined ordering over all records. \label{line:random-trunc}
\EndIf
\State $left\_df \gets left\_df.append(\mathcal{R})$
\EndFor

\For{$a \in domain(A)$}
\Comment{The domain for $A$ is taken from right\_df}
\State $\mathcal{R} \gets \{row \in right\_df | row[A] = a\}$
\If{$|\mathcal{R}| > 1$}
\State $right\_df \gets right\_df.drop(\mathcal{R})$ \label{line:drop-all-trunc}
\EndIf
\EndFor

\State \textbf{return:} left\_df.join(right\_df, on=A)
\Comment{Inner join}

\EndProcedure

\end{algorithmic}
\end{algorithm}

\begin{algorithm}[H]
\caption{\label{alg:vector-tabulate-pop-group} Subroutine of PHSafe which maps each row to the appropriate data cell and returns a vector of counts.}
\begin{algorithmic}[1]
\Require $df$: A private dataframe. This dataframe should only contain the records in the population group $P$.
\Require $P$: The population group.
\Require $f$: A function that maps each row to the appropriate data cell.

\Procedure{VectorizePopulationGroup}{$df, P,f$}
\State \textit{//Produce a Vector of Counts for Population Group $P$}

\State $v\leftarrow$ df.map($f$).groupby(PopGroup, Basis).count() 
\State \textbf{Return} $v$
\EndProcedure
\end{algorithmic}
\end{algorithm}

\begin{algorithm}[H]
\caption{\label{alg:vector-base-discrete-gaussian} The base discrete Gaussian mechanism.}
\begin{algorithmic}[1]
\Require $V$: An $n$ dimensional vector of integers.
\Require $\rho$: A privacy-loss parameter.
\Require $\Delta$: Stability factor
\Procedure{VectorDiscreteGaussian}{$V, \rho$, $\Delta$}
\State $y \gets \mathcal{N}^{n}_{\mathbb{Z}}\left(\frac{\Delta^2}{2\rho}\right)$
\Comment Defined in Section \ref{sec:privacy-prelim}
\State \textbf{return} $V + y$
\EndProcedure
\end{algorithmic}
\end{algorithm}

\section{Privacy Preliminaries}
\label{sec:privacy-prelim}

In this section, we give necessary background on zCDP, including its formal definition and relevant privacy properties.
\subsection{Privacy definitions}
\label{sec:privacy-defintions}

\begin{definition}[Neighboring Databases]
\label{def:neighboring-databases}
Let $x,x'$ be databases represented as multisets of tuples. We say that $x$ and $x'$ are \emph{neighbors} if their symmetric difference is 1.
\end{definition}

\begin{definition}[Bounded-Neighboring Databases]
\label{def:bounded-neighboring-databases}
Let $x,x'$ be databases represented as multisets of tuples. We say that $x$ and $x'$ are \emph{bounded neighbors} if they differ by \emph{arbitrarily changing} at most one tuple.
\end{definition}
We sometimes refer to neighboring databases as ``unbounded neighbors'' to differentiate between Definitions~\ref{def:neighboring-databases} and \ref{def:bounded-neighboring-databases}. We explicitly specify ``bounded'' when applicable and otherwise assume ``neighbors'' refers to Definition~\ref{def:neighboring-databases}.

We now define zCDP, which bounds the \emph{R\'enyi divergence} between the distributions of a mechanism run on neighboring databases.

\begin{definition}
The \emph{R\'enyi divergence of order $\alpha$} between distribution $P$ and distribution $Q$, denoted $D_{\alpha}(P \| Q)$, is defined as
\begin{equation}
    D_{\alpha}(P \| Q) = \frac{1}{\alpha-1}\log\left(\underset{x \sim P}{\mathbb{E}} \left[ \left( \frac{P(x)}{Q(x)} \right)^{\alpha-1}\right]\right)
\end{equation}

When $\alpha=\infty$,
\begin{equation}
    D_\infty(P\| Q) = \sup_{x \in supp(Q)} \log \left( \frac{P(x)}{Q(x)} \right)
\end{equation}
\end{definition}

\begin{definition}[zCDP \cite{BunS16}]
An algorithm $M: \mathcal{X} \rightarrow \mathcal{Y}$ satisfies $\rho$-zCDP if for all neighboring $x, x' \in \mathcal{X}$ and for all $\alpha \in (1, \infty)$,
\begin{equation}
    D_{\alpha}(M(x) \| M(x')) \le \rho \alpha.
\end{equation}
\end{definition}

We also define bounded zCDP, which considers bounded-neighboring databases instead of unbounded-neighboring databases.
\begin{definition}(Bounded zCDP \cite{BunS16})\label{def:bounded-zcdp}
An algorithm $M: \mathcal{X} \rightarrow \mathcal{Y}$ satisfies bounded $\rho$-zCDP if for all bounded neighbors $x, x' \in \mathcal{X}$ and for all $\alpha \in (1, \infty)$,
\begin{equation}
    D_{\alpha}(M(x) \| M(x')) \le \rho \alpha.
\end{equation}
\end{definition}

\subsection{Transformation stability}
A \emph{transformation} is a mapping from one dataset to another.

\begin{definition}[\cite{McSherry09}]
A $c$-stable transformation $T:\mathcal{X} \rightarrow \mathcal{Y}$ is a transformation such that:
\[|T(x) \ominus T(x')| \leq c \cdot (|x \ominus x'|)\] where $\ominus$ is the symmetric difference.
\end{definition}

\begin{lemma}[\cite{Gaboardi20}]
Let $T: \mathcal{X} \rightarrow \mathcal{Y}$ be a $b$-stable transformation and $G: \mathcal{Y} \rightarrow \mathcal{Z}$ be a $c$-stable transformation. Then $G \circ T: \mathcal{X} \rightarrow \mathcal{Z}$ is a $b\cdot c$-stable transformation. 
\end{lemma}

\begin{lemma}\label{lem:stab_cdp}
Let $T:\mathcal{X}\rightarrow \mathcal{Y}$ be a $b$-stable transformation. Let $M:\mathcal{Y}\rightarrow \mathcal{Z}$ be a $\rho$-zCDP mechanism. Then $M \circ T: \mathcal{X} \rightarrow \mathcal{Z}$ is a $b^2 \cdot \rho$-zCDP mechanism.
\end{lemma}

\begin{proof}
If $x$ and  $x'$ differ by one record, then $T(x)$ and $T(x')$ differ by at most $b$ records under $b$-stable transformation $T$. The result follows from the group privacy guarantee of zCDP \cite{BunS16}.
\end{proof}

\subsection{Base Mechanisms}
\label{sec:base-mechansisms}
\begin{definition}[L2 Sensitivity]
\label{def:sensitivity}
Given a vector function $q: \mathcal{X} \rightarrow \mathbb{Z}^n$, the L2 sensitivity of $q$ is \linebreak $\sup_{x \approx x'} \|q(x)-q(x')\|_2$ where $x \approx x'$ denotes $x$ and $x'$ are neighboring databases and $\| \cdot \|_2$ is the Euclidean norm.
\end{definition}
Like neighboring databases, there is an equivalent notion of bounded sensitivity.

\begin{definition}[Bounded L2 Sensitivity]
\label{def:sensitivity}
Given a vector function $q: \mathcal{X} \rightarrow \mathbb{Z}^n$, the bounded L2 sensitivity of $q$ is $\sup_{x \approx x'} \|q(x)-q(x')\|_2$ where $x \approx x'$ denotes $x$ and $x'$ are bounded neighboring databases and $\| \cdot \|_2$ is the Euclidean norm.
\end{definition}

\begin{definition}
\label{def:discrete-gaussian-distribution}
The discrete Gaussian distribution $\mathcal{N}_{\mathbb{Z}}(\sigma^2)$ centered at 0 is
\begin{equation}
\forall x \in \mathbb{Z}, \quad \Pr[X=x] = \frac{e^{-x^2/2\sigma^2}}{\sum_{y \in \mathbb{Z}}e^{-y^2/2 \sigma^2}}.
\end{equation}
\end{definition}

We can also consider the multidimensional discrete Gaussian distribution, $\mathcal{N}^n_{\mathbb{Z}}(\sigma^2)$ which is an $n$ length vector where each item in the vector is an independent sample from the discrete Gaussian distribution.

\begin{lemma}{\cite{CanonneK2020}}
\label{lem:vector-discrete-gaussian-satisfies-zcdp}
Let $q: \mathcal{X} \rightarrow \mathbb{R}^n$. Let $s$ be such that $\|q(x) - q(x')\|_2 \leq s$ for all neighboring $x, x'$. Let $\Delta > 0$. Then \textsc{VectorDiscreteGaussian}$(q(x), \rho,\Delta)$ from Algorithm~\ref{alg:vector-base-discrete-gaussian} satisfies $s^2\rho/\Delta^2$-zCDP as a function of $x$.
\end{lemma}

\subsection{Privacy Properties}

\subsubsection{Composition}
\label{sec:composition}

One of the most useful and important properties of privacy definitions is their behavior under composition.
In this section, we state sequential composition results for zCDP.

\begin{lemma}(Sequential composition of zCDP \cite{BunS16})
\label{lem:sequential-composition-zcdp}
Let $M_1: \mathcal{X} \rightarrow \mathcal{Y}$ and $M_2: \mathcal{X} \rightarrow \mathcal{Z}$ be mechanisms satisfying $\rho_1$-zCDP and $\rho_2$-zCDP respectively. Let $M_3(x) = (M_1(x), M_2(x))$. Then $M_3$ satisfies $(\rho_1 + \rho_2)$-zCDP.
\end{lemma}

\subsubsection{Postprocessing}
\label{sec:post-processing-background}

Another useful property is that zCDP is closed under postprocessing, meaning that the privacy guarantee cannot be weakened by manipulating the outputs of a zCDP mechanism without reference to the protected inputs.

\begin{lemma}(Postprocessing for zCDP \cite{BunS16})
\label{lem:post-processing}
Let $M: \mathcal{X} \rightarrow \mathcal{Y}$ and $f:\mathcal{Y} \rightarrow \mathcal{Z}$ be randomized algorithms. Suppose $M$ satisfies $\rho$-zCDP. Then $f \circ M: \mathcal{X} \rightarrow \mathcal{Z}$ satisfies $\rho$-zCDP.
\end{lemma}

\section{Privacy Analysis}\label{sec:privacy-analysis}
PHSafe measures privacy loss with respect to neighboring databases which differ in one person record. Under unbounded zCDP, this difference is the presence/absence (symmetric difference of 1) of one person record. Under bounded zCDP, this difference is an arbitrary change in a single person record.

Although PHSafe computations involve two separate dataframes, to avoid separate neighboring concepts (person vs. household), we imagine a model where both tables are produced by a transformation of an underlying single dataset called the base person dataframe. This hypothetical base person dataframe contains all necessary information to derive the person and unit dataframes described as inputs to the algorithm. The values of the person dataframe can be derived as a function of a single person record, resulting in a 1-stable transformation on the base person dataframe. The unit dataframe is produced by a function of multiple records, all of which share the same MAFID. This function is 2-stable because adding/removing one person can result in the original household being removed and another household replacing it.

\subsection{Stability}\label{sec:stability}
We now analyze the stability of the truncate and join step specified in Algorithm~\ref{alg:trunc-and-join}.
\begin{lemma}\label{lem:trunc-and-join-stability}
Algorithm \ref{alg:trunc-and-join} is a $(2\tau + 2)$-stable transformation.
\end{lemma}

\begin{wrapfigure}{r}{0.375\textwidth}
\begin{center}
    \includegraphics[width=0.35\textwidth]{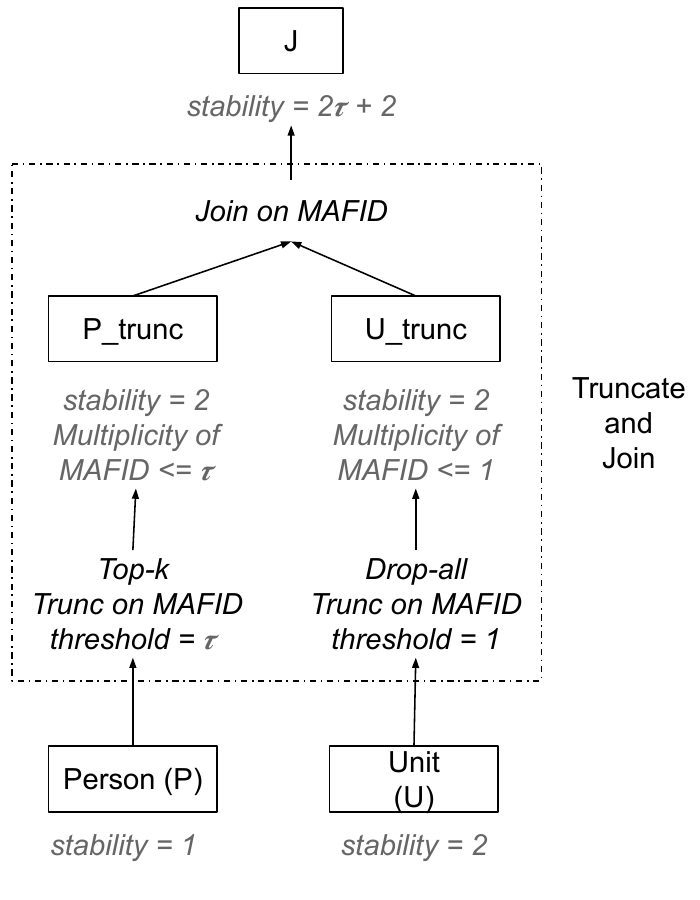}
  \end{center}
  \caption{Stability of the $J$ view. The person dataframe and unit dataframe are derived from the base person dataframe via a stability 1 and 2 transformation, respectively.}
\end{wrapfigure}

\begin{proof}
Let $P$ and $P'$ be neighboring person dataframes that differ in the presence or absence of one person. Let $U$ and $U'$ denote the corresponding unit dataframes. Let 
\begin{align*} 
J &=P.\text{truncate\_and\_join}(U, MAFID, \tau) \\ 
J' &=P'.\text{truncate\_and\_join}(U', MAFID, \tau) 
\end{align*} where the join occurs on the MAFID attribute. We show that $|J \ominus J'| \leq (2\tau+2)(|P \ominus P'|)=2\tau+2$.  Note, the truncate\_and\_join operator truncates each relation independently prior to the join.

First, we consider how adding or removing a record in the base person dataframe changes the two truncate\_and\_join operands post-truncation (in terms of symmetric difference). For the person dataframe, the truncation method in line \ref{line:random-trunc} takes the first $\tau$ records of each household according to a predefined ordering over the universe of records. This truncation is known to be $2$-stable \cite{Ebadi16}  since the addition of one record can result in the new record being included in the truncated set and another record being removed. That is, the post-truncation person dataframe differs by at most $2$ records prior to the join. For the unit dataframe, the truncation method in line \ref{line:drop-all-trunc} drops all records with non-unique MAFIDs. With respect to adding or removing a unit record from the unit dataframe, this is a $1$-stable transformation. However, the unit view itself has stability $2$ with respect to adding or removing a record in the base person dataframe. For example, removing a person record from a household can change the household type (e.g., from a married couple family to nonfamily) resulting in a symmetric difference of 2 for neighboring datasets. Hence, the post-truncation unit dataframe differs by at most $2$ records when adding or removing a record from the  base person dataframe. 

Next, we consider the join. Due to truncation, the max multiplicity of the MAFIDs of the person and unit dataframes is bounded by $\tau$ and $1$, respectively. Likewise, the stability of the truncated person and unit dataframe are both $2$.
The total stability of the view, $J$, is bounded by the MAFID multiplicity of the person dataframe times the stability of the unit dataframe plus the MAFID multiplicity of the unit dataframe times the stability of the person dataframe. This together results in a total change of at most $2\tau +2$.
\end{proof}

\subsection{PHSafe}
\label{sec:phsafe-gaussian-satisfies-zcdp}
Recall that, depending on the table in question, the PHSafe algorithm takes one of two forms: Algorithm~\ref{alg:phsafe-phjoin} or Algorithm~\ref{alg:phsafe-unit-count}. Here, we show how the privacy loss of each of these algorithms is derived. These values are with respect to unbounded zCDP, the conversion to bounded zCDP can be found in Section~\ref{sec:bounded-dp}.

\begin{theorem}
Algorithm \ref{alg:phsafe-phjoin} satisfies $\rho = \sum_{i \in [1,\omega]} \rho_i$-zCDP, where $\omega$ is the number of population group levels.
\end{theorem}

\begin{proof}
Let $P$ and $P'$ be neighboring databases, i.e., they differ in the presence or absence of one person. Let $U$ and $U'$ denote the corresponding unit dataframes. 
We want to show that for all $\alpha \in (1, \infty)$, 
    \[    D_{\alpha}(M(P, U) \| M(P', U')) \le \alpha \cdot \rho.\]
    where $M$ is Algorithm~\ref{alg:phsafe-phjoin}. 

We can decompose the algorithm $M$ into a transformation function $T$ followed by a sequential composition of $\omega$ zCDP mechanisms $M_1, \ldots, M_\omega$ that take as input the joined dataframes output by $T$ and output the statistics for population group level $i$. 

Let $J$ and $J'$ be the result of transformation $T$ on $(P, U)$ and $(P', U')$, respectively. Since the mappings, $f$ and $g_i$ are 1-stable transformations, the L2 distance in the query defined in Lines~\ref{line:query-start}-\ref{line:query-end} is  1. It follows that for all $i \in [1, \omega], \ D_{\alpha}(M_i(J) \| M_i(J')) \le \alpha \cdot \rho_i/(2\tau+2)^2$ by Lemma \ref{lem:vector-discrete-gaussian-satisfies-zcdp}. Then, by sequential composition, we have $D_{\alpha}(\{M_i(J)\}_i \| \{M_i(J')\}_i) \le \alpha \cdot \sum_i\rho_i/(2\tau+2)^2$.
Finally, since the stability of $T$ is $2\tau + 2$ (per Lemma \ref{lem:trunc-and-join-stability}) we get: \[D_{\alpha}\left((\{M_i\}_i\circ T) (P,U) \|(\{M_i\}_i \circ T) (P', U')\right) \le \alpha \cdot \sum_i\rho_i = \alpha \cdot \rho.\]
\end{proof}

For completeness, we also give privacy guarantees for Algorithm~\ref{alg:phsafe-unit-count}. The proof is straightforward, as this algorithm does not involve a join operator.
\begin{theorem}
Algorithm \ref{alg:phsafe-unit-count} is $\rho = \sum_{i \in [1,\omega]} \rho_i$-zCDP, where $\omega$ is the number of population group levels. 
\end{theorem}
\begin{proof}
Let $P$ and $P'$ be neighboring databases that differ in the presence or absence of one person. Let $U$ and $U'$ denote the corresponding unit dataframes.
We want to show that for all $\alpha \in (1, \infty)$,
    \[    D_{\alpha}(M(U) \| M(U')) \le \alpha \cdot \rho/2.\]
    where $M$ is Algorithm~\ref{alg:phsafe-unit-count}. 
    
Since the unit view has a stability of 2, the result then follows.

We can decompose the algorithm $M$ into a transformation function $T$ followed by a sequential composition of $\omega$ zCDP mechanisms $M_1, \ldots, M_\omega$ that take as input the transformed unit data frames output by $T$ and output the statistics for population group level $i$. 

Let $J$ and $J'$ be the result of transformations $T$ on $U$ and $U'$, respectively. Since the mappings, $f$ and $g_i$ are 1-stable transformations, the  L2 distance in the query defined in Lines~\ref{line:denom-query-start}-\ref{line:denom-query-end} is 1. It follows that for all $i \in [1, \omega], \ D_{\alpha}(M_i(J) \| M_i(J')) \le \alpha \cdot \rho_i/2$ by Lemma~\ref{lem:vector-discrete-gaussian-satisfies-zcdp}. Then, by sequential composition, we have $D_{\alpha}(\{M_i(J)\}_i \| \{M_i(J')\}_i) \le \alpha \cdot \sum_i\rho_i/2$.
Finally, since $T$ is a filter, its stability is $1$ so we get: \[D_{\alpha}(\{M_i\}_i \circ T (U) \| \{M_i\}_i \circ T (U')) \le \alpha \cdot \sum_i\rho_i/2.\] 
\end{proof}

\subsection{Converting to Bounded Differential Privacy}
\label{sec:bounded-dp}
Here, we convert the previous privacy guarantees from unbounded zCDP to bounded zCDP. While all unbounded $\rho$-zCDP mechanisms also satisfy bounded $4\rho$-zCDP, we demonstrate here that Algorithm~\ref{alg:phsafe-phjoin} and Algorithm~\ref{alg:phsafe-unit-count} instead satisfy bounded $2\rho$-zCDP.

\begin{theorem} \label{thm:bounded_zcdp_alg1}
Algorithm \ref{alg:phsafe-phjoin} satisfies bounded $\rho = \sum_{i \in [1,\omega]} 2\rho_i$-zCDP.
\end{theorem}

\begin{proof}
Let $P$ and $P'$ be bounded neighboring databases, i.e., they differ in an arbitrary change in one person's record. Let $U$ and $U'$ denote the corresponding unit dataframes.
We want to show that for all $\alpha \in (1, \infty)$, 
    \[    D_{\alpha}(M(P, U) \| M(P', U')) \le \alpha \cdot 2\rho\]
    where $M$ is Algorithm~\ref{alg:phsafe-phjoin}. 

When analyzing under bounded zCDP, it is sufficient to consider the deletion of a single individual’s record and the addition of another arbitrary record. We can decompose the algorithm $M$ into a transformation function $T$ followed by a sequential composition of $\omega$ zCDP mechanisms $M_1, \ldots, M_\omega$ that take as input the joined dataframes output by $T$ and output the statistics for population group level $i$.  

By Lemma~\ref{lem:trunc-and-join-stability}, Algorithm \ref{alg:trunc-and-join} is a $(2\tau + 2)$-stable transformation. Let $J$ and $J'$ be the result of transformations $T$ on $(P, U)$ and $(P', U')$, respectively. Since Algorithm~\ref{alg:trunc-and-join} is a $(2\tau +2)$-stable transformation and the mappings $f$ and $g_i$ are 1-stable transformations, the composition of all the transformations is $(2\tau +2)$-stable. Therefore, the maximum difference in the query defined in Lines~\ref{line:query-start}-\ref{line:query-end} for a single cell due to adding a record is an increase by $(2\tau +2)$. Likewise, for removing a record, the maximum difference in the query defined in Lines~\ref{line:query-start}-\ref{line:query-end} is a decrease in a single cell by $(2\tau +2)$. This results in a maximum L2 distance of $\sqrt{2(2\tau +2)^2}= \sqrt{2}(2\tau +2)$.

It follows that, given a privacy-loss parameter of $\frac{\rho_i}{(2\tau +2)^2}$, for all $i \in [1, \omega]$, we have $\D_{\alpha}(M_i(J) \| M_i(J')) \le \alpha \cdot 2\rho_i$ by Lemma \ref{lem:vector-discrete-gaussian-satisfies-zcdp}. Then, by sequential composition, we have $D_{\alpha}(\{M_i(J)\}_i \| \{M_i(J')\}_i) \le \alpha \cdot \sum_i2\rho_i =\alpha \cdot 2\rho$.
\end{proof}

For completeness, we also give privacy guarantees for Algorithm~\ref{alg:phsafe-unit-count}. The proof is straightforward, as this algorithm does not involve a join operator.

\begin{theorem} \label{thm:bounded_zcdp_alg2}
Algorithm \ref{alg:phsafe-unit-count} satisfies bounded $\rho =\sum_{i \in [1,\omega]} 2\rho_i$-zCDP.
\end{theorem}
\begin{proof}
Let $P$ and $P'$ be bounded neighboring databases that differ in an arbitrary change in one person's record. Let $U$ and $U'$ denote the corresponding unit dataframes.
We want to show that for all $\alpha \in (1, \infty)$,
    \[    D_{\alpha}(M(U) \| M(U')) \le \alpha \cdot 2\rho.\]
where $M$ is Algorithm~\ref{alg:phsafe-unit-count}. 
    
When analyzing under bounded zCDP it is sufficient to consider the deletion of a single individual’s record and the addition of another arbitrary record. We can decompose the algorithm $M$ into a transformation function $T$ followed by a sequential composition of $\omega$ zCDP mechanisms $M_1, \ldots, M_\omega$ that take as input the transformed unit data frames output by $T$ and output the statistics for population group level $i$. 

Let $J$ and $J'$ be the result of transformations $T$ on $U$ and $U'$, respectively. Since the unit dataframe is a $2$-stable transformation and the mappings $f$ and $g_i$ are 1-stable transformations, the composition of all the transformations is $2$-stable. Therefore, the maximum difference in the query defined in Lines~\ref{line:denom-query-start}-\ref{line:denom-query-end} for a single population group due to adding a record is an increase by $2$. Likewise, for removing a record, the maximum difference in the query defined in Lines~\ref{line:denom-query-start}-\ref{line:denom-query-end} is a decrease in a single population group by $2$. This results in a maximum L2 distance of $\sqrt{2(2)^2}= \sqrt{2}(2)$. 

It follows that, given a privacy parameter of $\frac{\rho_i}{4}$, for all $i \in [1, \omega]$, we have $\D_{\alpha}(M_i(J) \| M_i(J')) \le \alpha \cdot 2\rho_i$ by Lemma \ref{lem:vector-discrete-gaussian-satisfies-zcdp}. Then, by sequential composition, we have $D_{\alpha}(\{M_i(J)\}_i \| \{M_i(J')\}_i) \le \alpha \cdot \sum_i 2\rho_i =\alpha \cdot 2\rho$.

\end{proof}

\section{Implementation of PHSafe}\label{sec:implementation}

The PHSafe algorithm pseudocode from this paper and the implementation of PHSafe differ in some details. This section highlights a selection of implementation differences, focusing on the aspects relevant to the privacy analysis of PHSafe. In particular, we argue why each of these differences do not materially impact the proven privacy-loss budgets from Section \ref{sec:privacy-analysis}.

\subsection{Input Validation}\label{sec:validation}

PHSafe validates its inputs to ensure that they conform to expected schema specifications and that they are internally consistent. Input validation is necessary for successful deployment of large-scale systems like the 2020 DAS. However, this validation process does not fit into the differential privacy framework. For instance, input validation errors may necessitate edits to the confidential input data or require a re-run of the program. Such activities are not captured by the privacy guarantees of the algorithm. We view the privacy risks associated with input validation as negligible since the PHSafe program is being deployed by a trusted curator, the Census Bureau. Hence, potential risks of input validation fall outside the assumed operational threat model and have no material impact on the privacy loss incurred by a successful production run of the algorithm.

\subsection{Tumult Analytics}\label{sec:analytics}

PHSafe is implemented using Tumult Analytics \cite{berghel2022tumult}, a Python library for writing and executing differentially private queries. Tumult Analytics automates stability and noise scale calculations, avoiding potential pitfalls with manually calibrating the noise distributions. All access to the private data is managed through a Tumult Analytics \texttt{Session}, which tracks and limits the total privacy loss of computations on the sensitive data. The \texttt{Session} initializes all transformations and measurements performed on the private data. PHSafe constructs a \texttt{Session} with: 
\begin{itemize}
    \item The total privacy-loss budget for all queries made on the private dataset.
    \item The private datasets of person and unit records.
    \item The privacy definition for the program (e.g., zCDP).
    \item Each private dataset's neighboring definition. The person dataset defines neighbors with respect to the addition/removal of any single record from the dataset. The household dataset defines neighbors with respect to the addition/removal of two records from the dataset. These setting reflect the use of unbounded neighbors within the PHSafe code implementation.
\end{itemize}

The pseudocode of PHSafe describes an algorithm that is re-run for each table it produces. The PHSafe program encodes the processing of all tables into a single run so that all processes can be contained within a single \texttt{Session}. This distinction does not change the privacy analysis of PHSafe.

\subsubsection{Private Joins}
Algorithm \ref{alg:trunc-and-join}, the truncate and join operator, is implemented as a Tumult Analytics \texttt{join\_private} operation. Truncation strategies for both the left (persons) and right (units) dataframes are specified as part of the \texttt{join\_private} operation.

The \texttt{join\_private} in PHSafe uses Tumult Analytics' \texttt{DropExcess($\tau$)} strategy for the left dataframe. This strategy keeps $\tau$ records for each join key and drops the rest, as described in the pseudocode. To determine which records are kept, the records are ordered by a hash of the record value\footnote{Additional steps are taken to ensure identical records are hashed to different values.}.

The \texttt{join\_private} uses the \texttt{DropNonUnique} strategy for the right dataframe, which has the behavior described in Algorithm \ref{alg:trunc-and-join} -- it drops any records with non-unique join values. This is expected to be a no-op in practice (all units have unique IDs), and it ensures the tightest possible stability analysis in this scenario, as described in Section \ref{sec:stability}.

\subsection{Preprocessing}

The pseudocode for PHSafe assumed simplified input schemas compared to the actual inputs of the PHSafe implementation. The PHSafe implementation involves several preprocessing data transformations. However, these transformations do not change the stability analysis of the assumed input dataframes of the pseudocode. For example, the pseudocode excluded the group quarters population from its input data but, in reality, this exclusion requires a filter transformation on the 2020 CEF. Because filter transformations have a stability of one, the privacy analysis from Section \ref{sec:privacy-analysis} is unchanged.

\subsection{Postprocessing}

\subsubsection{PH5 and PH8 Numerator Calculations}

As previously discussed, PH5\_num and PH8\_num are not directly estimated as part of the differential privacy routine of PHSafe. They are created from the estimates PHSafe produces for PH4 and PH7, respectively. Under differential privacy, these activities are categorized as postprocessing and, therefore, do not incur additional privacy loss. 

\subsubsection{Statistical Postprocessing}
The output of the PHSafe program is passed through a statistical postprocessing algorithm before being handed off to the Decennial Tabulation System. The statistical postprocessing algorithm's goals are to ensure certain demographic reasonableness of the S-DHC and provide credible intervals. As part of this effort, statistical postprocessing ensures nonnegativity of counts and averages. In addition, it ensures both the numerator and denominator do not yield negative or infinite (extremely large) ratios. Finally, it provides credible intervals -- or the 90\% probability that the enumerated value is between the lower and upper end points of the interval. Credible intervals encompass the noise infused by disclosure avoidance; they do not provide estimates of expected error introduced by truncation or non-disclosure avoidance error (e.g., coverage error). Because the statistical postprocessing effort does not utilize the 2020 CEF, it fits the assumptions of the differential privacy postprocessing theorem. Therefore, it does not change the privacy analysis of the PHSafe algorithm.    

\subsection{Variance tracking}
In addition to noisy count estimates for population groups, the PHSafe implementation also outputs the variance of the noise distributions used to generate each estimate. One benefit of differential privacy as a disclosure avoidance technique is the ability to release properties such as variance without compromising the privacy guarantees. Knowing the variance and type of distribution (discrete Gaussian) is crucial to the success of statistical postprocessing.   

\subsection{Puerto Rico}\label{sec:pr}

As a matter of Census Bureau policy, the privacy-loss accounting for Puerto Rico is handled separately from the privacy loss of the United States. Hence, the Puerto Rico data records are processed in a separate execution of the PHSafe algorithm. One notable difference with PHSafe's processing of Puerto Rico is the lack of a Nation geography level and corresponding population group levels. Otherwise, PHSafe follows the same steps for Puerto Rico as it does for the United States.

\section{Parameters and Tuning}\label{sec:params}

Throughout this paper, we have referenced a number of parameters required by PHSafe. Parameters are tunable inputs that are necessary to fully determine the nature of the program (e.g., the noise scale employed in \textsc{VectorDiscreteGaussian}, the privacy-loss budgets for population group levels, and truncation thresholds). Parameter specification is a matter of policy. The Census Bureau's Data Stewardship Executive Policy (DSEP) committee approved all production parameters for the S-DHC. Parameter selection necessitates trade-offs, as many of these parameters are dependent on each other. To illustrate, we consider a fundamental relationship between the privacy-loss parameters and their corresponding margins of error with discrete Gaussian noise distributions.

\subsection{Error bounds}\label{sec:error-bounds}
PHSafe was designed to have predictable error due to noise infusion. Truncation in PHSafe introduces an additional source of error that is not predictable. Nonetheless, the Census Bureau uses margins of error to set accuracy targets for the noise infusion component of PHSafe. 

\begin{definition}
The 90\% MOE is half the width of the 90\% confidence interval.
\end{definition}

Since we only consider 90\% MOEs in this paper, we often write MOE without the 90\% qualifier. We begin by stating a portion of Proposition 25 from \cite{CanonneK2020}.

\begin{proposition}[Proposition 25 in \cite{CanonneK2020}]\label{lem:discrete-gaussian-bound}

For all $m \in \mathbb{Z}$ with $m \geq 1$, and for all $\sigma \in \mathbb{R}$ with $\sigma > 0$, $\Pr[X \geq m]_{X \leftarrow \mathcal{N}_{\mathbb{Z}}(\sigma^2)} \leq \Pr[X \geq m-1]_{X \leftarrow \mathcal{N}(\sigma^2)}$.
\end{proposition}

That is, discrete Gaussian distributions have tighter tails than their corresponding continuous Gaussian distributions. The following corollary reinterprets the tail bounds for real-valued numbers. 
\begin{corollary}
For all $m, \sigma \in \mathbb{R}$ with $x \geq 1$ and $\sigma > 0$,  $\Pr[X > x]_{X \leftarrow \mathcal{N}_{\mathbb{Z}}(\sigma^2)} \leq \Pr[X > \lfloor x \rfloor]_{X \leftarrow \mathcal{N}(\sigma^2)}$.
\end{corollary}

Figure 2 of \cite{CanonneK2020} provides an intuitive visualization of these tail bounds.
Using the continuous Gaussian to upper bound MOE in the discrete Gaussian, it follows that the discrete Gaussian has $X \in [-\lfloor 1.64\sigma \rfloor, \lfloor 1.64\sigma \rfloor]$ with probability of at least 90\%.
That is, $MOE \leq \lfloor 1.64\sigma \rfloor$. 

Recall that $\sigma^2 = \frac{\Delta^2}{2\rho}$ in Algorithm \ref{alg:vector-base-discrete-gaussian}. Combining these two equations and solving for $\rho$, yields the following result.

\begin{corollary}\label{cor:dgauss-rho-from-moe}
For any $\Delta > 0$, the base discrete Gaussian mechanism run with $(\rho = \frac{1.3448\Delta^2}{\lfloor MOE \rfloor^2}, \Delta)$ has 90\% margins of error of at most $MOE$ in each vector component.
\end{corollary}

For any table using Algorithm \ref{alg:phsafe-phjoin}, a basis table cell in population group level $i$ has an MOE of $\left\lfloor 1.64\sqrt{\frac{(2\tau + 2)^2}{2\rho_i}}\right\rfloor$.

For any table using Algorithm \ref{alg:phsafe-unit-count}, a basis table cell in population group level $i$ has an MOE of $\left\lfloor 1.64\sqrt{\frac{2}{\rho_i}}\right\rfloor$.


\subsection{Parameter Identification, Trade-offs, and Outcomes}

Before delving into specific parameters, we give an overview of some of the methods involved in navigating the general trade-offs associated with setting parameters. Firstly, parameters tend to impact some combination of these three aspects: data confidentiality, data accuracy, and data availability. Data availability refers to the volume of tabular statistics released. Data confidentiality reflects privacy as measured by the privacy-loss budgets of the algorithm. For our purposes, data accuracy is measured by margins of error on the tabular statistics. For example, excluding population group level $i$ would reduce data availability but improve privacy, since the privacy-loss budget $\rho_i$  is no longer necessary for each table. To aid in the understanding of these trade-offs, we relied on a combination of tangible tools and theoretical analyses.  

We created an analysis tool, PHExplore, to provide hands-on experience exploring these trade-offs. We provide a brief summary of this tool in the following section. Then, we highlight specific parameters and the critical decisions made by the DSEP committee based on recommendations from subject-matter expert or DAS scientists. Several of the subject-matter expert recommendations were influenced by interaction with the PHExplore tool.

\subsection{Parameter Tuning Using the PHExplore}

PHExplore is an easy-to-use interactive decision support tool implemented in the Microsoft Excel program and developed to facilitate conversations between subject-matter experts, disclosure avoidance scientists, and ultimately, the DSEP committee. 

The tool allowed users to interactively specify: 
\begin{itemize}
\item The set of geography levels and iteration levels that constitute the universe of population groups for which statistics are tabulated. 
\item The truncation thresholds that limit the maximum number of people per household.
\item Privacy-loss budgets.
\end{itemize}
Based on these parameters, the tool computed the expected MOEs of the algorithm. The computations were performed using analytical formulae for expected error of noise mechanisms employed in the PHSafe algorithm. 

In the next section, we describe  some of the key parameters considered and the decision process used to set these parameters. 

\subsubsection{Parameter Selection}

\noindent \emph{\textbf{Population group levels}}: As a reminder, a population group level is defined by a geographic summary level, such as Nation or State, and an iteration level (i.e., Unattributed, A-G, or H-I). PHExplore helped subject-matter experts grasp the impact that adding or removing levels would have on the privacy-loss budget, but they had to weigh that against the value of having publishable statistics at each given level. They also gathered feedback from data users on these topics. Given privacy-loss budget constraints, they ultimately settled on the levels referenced in Table \ref{tab:moe-targets}.

\begin{table}[t]
    \centering
    \resizebox{\columnwidth}{!}{
    \begin{tabular}{c c c c c c}
    \toprule
    Table & Truncation Threshold & Population Group Level & MOE Target &  Unbounded Privacy Loss & Bounded Privacy Loss \\
    \midrule
    \multirow{6}{*}{PH1\_num} & \multirow{6}{*}{10} & (Nation, Unattributed) & 500 & 0.002619 & 0.005238 \\
                              &                     & (Nation, A-G)          & 500 & 0.002619 & 0.005238 \\
                              &                     & (Nation, H-I)          & 500 & 0.002619 & 0.005238 \\
                              &                     & (State, Unattributed)  & 200 & 0.016371 & 0.032742 \\
                              &                     & (State, A-G)           & 68 & 0.141622 & 0.283244 \\
                              &                     & (State, H-I)           & 200 & 0.016371 & 0.032742 \\
    \midrule
    \multirow{6}{*}{PH1\_denom} & \multirow{6}{*}{NA} & (Nation, Unattributed) & 500 & 0.000022 & 0.000044 \\
                                &                     & (Nation, A-G)          & 500 & 0.000022 & 0.000044 \\
                                &                     & (Nation, H-I)          & 500 & 0.000022 & 0.000044 \\
                                &                     & (State, Unattributed)  & 200 & 0.000135 & 0.00027  \\
                                &                     & (State, A-G)           & 68 & 0.00117 & 0.00234  \\
                                &                     & (State, H-I)           & 200 & 0.000135 & 0.00027  \\
    \midrule
    \multirow{2}{*}{PH2} & \multirow{2}{*}{10} & (Nation, Unattributed) & 500 & 0.002619 & 0.005238 \\                                    
                         &                      & (State, Unattributed)  & 200 & 0.016371 & 0.032742 \\
    \midrule
    \multirow{6}{*}{PH3} & \multirow{6}{*}{6} & (Nation, Unattributed) & 500 & 0.001061 & 0.002122 \\
                              &                     & (Nation, A-G)          & 500 & 0.001061 & 0.002122 \\
                              &                     & (Nation, H-I)          & 500 & 0.001061 & 0.002122 \\
                              &                     & (State, Unattributed)  & 200 & 0.006630 & 0.01326 \\
                              &                     & (State, A-G)           & 20 & 0.662976 & 1.325952 \\
                              &                     & (State, H-I)           & 200 & 0.006630 & 0.01326 \\
    \midrule
    \multirow{6}{*}{PH4} & \multirow{6}{*}{10} & (Nation, Unattributed) & 500 & 0.002619 & 0.005238 \\
                              &                     & (Nation, A-G)          & 500 & 0.002619 & 0.005238 \\
                              &                     & (Nation, H-I)          & 500 & 0.002619 & 0.005238 \\
                              &                     & (State, Unattributed)  & 200 & 0.016371 & 0.032742 \\
                              &                     & (State, A-G)           & 68 & 0.141622 & 0.283244 \\
                              &                     & (State, H-I)           & 200 & 0.016371 & 0.032742 \\
    \midrule
    \multirow{6}{*}{PH5\_denom} & \multirow{6}{*}{NA} & (Nation, Unattributed) & 500 & 0.000022 & 0.000044 \\
                                &                     & (Nation, A-G)          & 500 & 0.000022 & 0.000044 \\
                                &                     & (Nation, H-I)          & 500 & 0.000022 & 0.000044 \\
                                &                     & (State, Unattributed)  & 200 & 0.000135 & 0.00027  \\
                                &                     & (State, A-G)           & 68 & 0.00117 & 0.00234  \\
                                &                     & (State, H-I)           & 200 & 0.000135 & 0.00027  \\
    \midrule
    \multirow{2}{*}{PH6} & \multirow{2}{*}{6}  & (Nation, Unattributed) & 500 & 0.001061 & 0.002122 \\
                         &                     & (State, Unattributed)  & 200 & 0.006630 & 0.01326 \\
    \midrule
    \multirow{6}{*}{PH7} & \multirow{6}{*}{10} & (Nation, Unattributed) & 500 & 0.002619 & 0.005238 \\
                              &                     & (Nation, A-G)          & 500 & 0.002619 & 0.005238 \\
                              &                     & (Nation, H-I)          & 500 & 0.002619 & 0.005238 \\
                              &                     & (State, Unattributed)  & 200 & 0.016371 & 0.032742 \\
                              &                     & (State, A-G)           & 68 & 0.141622 & 0.283244 \\
                              &                     & (State, H-I)           & 200 & 0.016371 & 0.032742 \\
    \midrule
    \multirow{6}{*}{PH8\_denom} & \multirow{6}{*}{NA} & (Nation, Unattributed) & 500 & 0.000022 & 0.000044 \\
                                &                     & (Nation, A-G)          & 500 & 0.000022 & 0.000044 \\
                                &                     & (Nation, H-I)          & 500 & 0.000022 & 0.000044 \\
                                &                     & (State, Unattributed)  & 200 & 0.000135 & 0.00027  \\
                                &                     & (State, A-G)           & 68 & 0.00117 & 0.00234  \\
                                &                     & (State, H-I)           & 200 & 0.000135 & 0.00027  \\
    \midrule
    \bottomrule
    \end{tabular}
    }
   \caption{MOE targets for the statistics released at different population group levels, along with the corresponding privacy loss (unbounded and bounded $\rho$-zCDP for discrete Gaussian). Due to the total population of the United States being published without noise, the bounded privacy-loss budgets in this table were stressed by Census Bureau staff in internal conversations and for presentation to DSEP, for purposes of interpreting the privacy guarantee.}
   \label{tab:moe-targets}
\end{table}

\noindent \emph{\textbf{MOE, $\rho$, and truncation}}:  PHExplore provided an interface for adjusting privacy loss and truncation thresholds to observe the impact on expected MOE as derived in Section \ref{sec:error-bounds}. The Census Bureau selected the MOEs and corresponding $\rho$s and truncation thresholds as displayed in Table \ref{tab:moe-targets} for the production run of PHSafe on the 2020 Census data.

\section{Conclusion}\label{sec:discussion}

In this paper, we presented the PHSafe algorithm, a differentially private algorithm for producing the S-DHC for the Census Bureau. We covered several key aspects of the algorithm. First, we provided a technical pseudocode description of the algorithm. Then, we covered the privacy properties of the algorithm. Next, we discussed the differences between the pseudocode and the implementation of the algorithm with Tumult Analytics. We also covered PHExplore's role in parameter tuning and the key algorithmic parameters set by Census Bureau policy.

\bibliographystyle{unsrt}
\bibliography{refs}

\begin{thebibliography}{1}

\bibitem{title13}
{U.S. Code Title 13—Census}.
\newblock \url{https://www.law.cornell.edu/uscode/text/13}.

\bibitem{ap-census-attack}
Seth Borenstein.
\newblock {Potential Privacy Lapse Found in Americans’ 2010 Census Data}.
\newblock \url{https://apnews.com/article/aba8e57c145047b5bab11b62baaa7f7a}, February 2019.

\bibitem{dsep-dp}
{Memorandum 2019.25: 2010 Demonstration Data Products – Design Parameters and Global Privacy-Loss Budget}.
\newblock \url{https://www.census.gov/programs-surveys/decennial-census/2020-census/planning-management/memo-series/2020-memo-2019_25.html}, October 2019.

\bibitem{BunS16}
Mark Bun and Thomas Steinke.
\newblock {Concentrated Differential Privacy: Simplifications, Extensions, and Lower Bounds}.
\newblock {\em CoRR}, abs/1605.02065, 2016.

\bibitem{McSherry09}
Frank McSherry.
\newblock {Privacy Integrated Queries: An Extensible Platform for Privacy-preserving Data Analysis}.
\newblock In Ugur {\c{C}}etintemel, Stanley~B. Zdonik, Donald Kossmann, and Nesime Tatbul, editors, {\em Proceedings of the {ACM} {SIGMOD} International Conference on Management of Data, {SIGMOD} 2009, Providence, Rhode Island, USA, June 29 - July 2, 2009}, pages 19--30. {ACM}, 2009.

\bibitem{Gaboardi20}
Marco Gaboardi, Michael Hay, and Salil Vadhan.
\newblock {A Programming Framework for OpenDP}.
\newblock 2020.

\bibitem{CanonneK2020}
Cl{\'{e}}ment~L. Canonne, Gautam Kamath, and Thomas Steinke.
\newblock {The Discrete Gaussian for Differential Privacy}.
\newblock {\em CoRR}, abs/2004.00010, 2020.

\bibitem{Ebadi16}
Hamid Ebadi, Thibaud Antignac, and David Sands.
\newblock {Sampling and Partitioning for Differential Privacy}.
\newblock In {\em 2016 14th Annual Conference on Privacy, Security and Trust (PST)}, pages 664--673, 2016.

\bibitem{berghel2022tumult}
Skye Berghel, Philip Bohannon, Damien Desfontaines, Charles Estes, Sam Haney, Luke Hartman, Michael Hay, Ashwin Machanavajjhala, Tom Magerlein, Gerome Miklau, Amritha Pai, William Sexton, and Ruchit Shrestha.
\newblock {Tumult Analytics: A Robust, Easy-to-use, Scalable, and Expressive Framework for Differential Privacy}.
\newblock \url{https://arxiv.org/abs/2212.04133}, 2022.

\end{thebibliography}

\end{document}